\documentclass{article}
\usepackage{ijcai23}

\usepackage{times}
\usepackage{soul}
\usepackage{url}
\usepackage[hidelinks]{hyperref}
\usepackage[utf8]{inputenc}
\usepackage[small]{caption}
\usepackage{subcaption}
\usepackage{graphicx}
\usepackage{booktabs}
\usepackage{algorithm}
\usepackage{algorithmic}
\usepackage[switch]{lineno}
\usepackage{multirow}
\usepackage{tabularx}
\usepackage{amsmath}
\newtheorem{theorem}{Theorem}[section]
\newtheorem{corollary}{Corollary}[theorem]
\newtheorem{lemma}[theorem]{Lemma}
\usepackage{amsfonts}
\newenvironment{proof}{\paragraph{Proof:}}{\hfill$\square$}

\newcommand{\funcname}[1]{{\texttt{#1}}}

\title{Doubly Stochastic Graph-based Non-autoregressive Reaction Prediction}

\author{
Ziqiao Meng $^{1}$\thanks{This work is done when Ziqiao Meng worked as an intern in Tencent AI Lab.}
\and Peilin Zhao $^{2}$\thanks{Corresponding authors: Peilin Zhao, and Irwin King.}
\and Yang Yu $^{2}$
\And Irwin King $^{1}$\footnotemark[2]
\affiliations
$^1$The Chinese University of Hong Kong\\
$^2$Tencent AI Lab\\
\emails
\{zqmeng, king\}@cse.cuhk.edu.hk,
mazonzhao@tencent.com,
kevinyyu@tencent.com
}

\begin{document}

\maketitle

\begin{abstract}
Organic reaction prediction is a critical task in drug discovery. Recently, researchers have achieved non-autoregressive reaction prediction by modeling the redistribution of electrons, resulting in state-of-the-art top-1 accuracy, and enabling parallel sampling. However, the current non-autoregressive decoder does not satisfy two essential rules of electron redistribution modeling simultaneously: the electron-counting rule and the symmetry rule. This violation of the physical constraints of chemical reactions impairs model performance. In this work, we propose a new framework called \textit{ReactionSink} that combines two doubly stochastic self-attention mappings to obtain electron redistribution predictions that follow both constraints. We further extend our solution to a general multi-head attention mechanism with augmented constraints. To achieve this, we apply Sinkhorn's algorithm to iteratively update self-attention mappings, which imposes doubly conservative constraints as additional informative priors on electron redistribution modeling. We theoretically demonstrate that our \textit{ReactionSink} can simultaneously satisfy both rules, which the current decoder mechanism cannot do. Empirical results show that our approach consistently improves the predictive performance of non-autoregressive models and does not bring an unbearable additional computational cost.
\end{abstract}

\section{Introduction}

\begin{figure}[ht]
    \centering
    \includegraphics[width=0.48\textwidth]{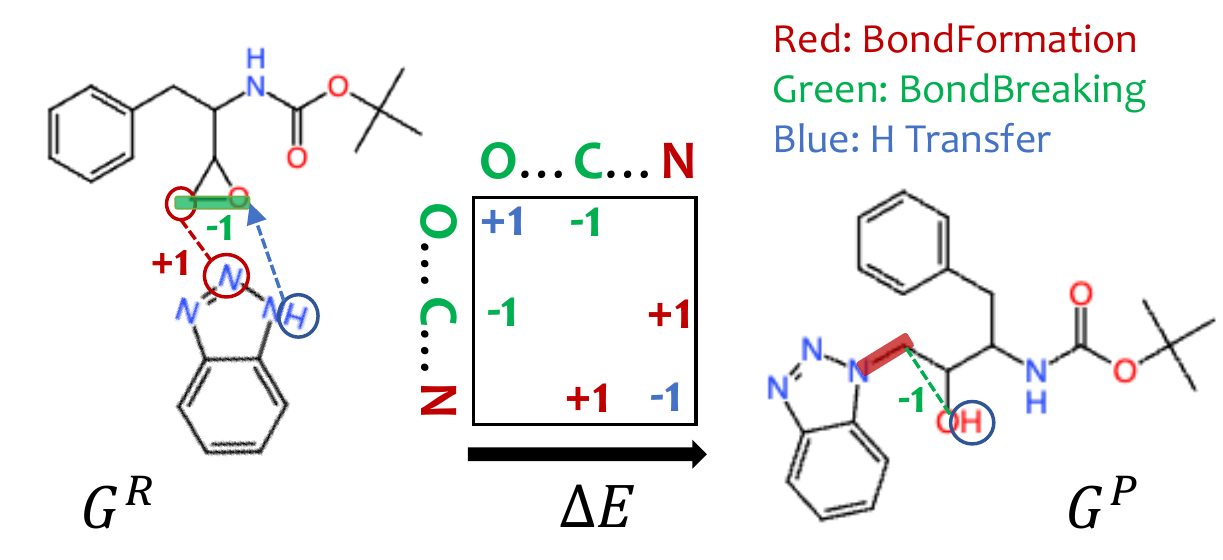}
    \caption{This figure illustrates chemical reaction from electron redistribution perspective. Red circles and arrows denote involved atoms that form new bond. Green arrow denotes the bond that breaks. Blue circles and arrows describe the hydrogen atom transfer. With electron redistribution $\Delta E$, reactants $G^{R}$ are transformed to products $G^{P}$ such that $E^{R}+\Delta E = E^{P}$, where $E^{R}$ and $E^{P}$ are the adjacency matrices of reactants and products respectively. Note that $E^{R}$, $E^{P}$ and $\Delta E$ does not consider Hydrogen atoms explicitly and $\Delta E$ records the change of attached hydrogen atoms in diagonal.}
    \label{fig:electron-transfer}
\end{figure}

Reaction prediction is a crucial task in computational chemistry. With a reliable prediction model, chemists can verify the accuracy of retrosynthetic routes and potentially uncover new chemical reactions, which could significantly benefit the drug discovery industry. However, given a set of reactants, the number of potential products can be combinatorially large, necessitating automatic tools to narrow down the search space. Consequently, the use of deep learning techniques to automatically infer possible reaction products has become prevalent and essential.

Translation-based autoregressive models \cite{moleculartransformer,Tetko,Graph2smiles,ReactionT5,chemformer,motif-reaction} have dominated the design of template-free end-to-end reaction models in recent years. These models represent molecules as sequences of characters, called SMILES \cite{SMILES}, and formulate the reaction prediction problem as a neural machine translation problem \cite{seq2seq,transformer,BERT}. This class of models achieves great predictive performance and does not rely on pre-computed atom-mapping information. However, autoregressive models have two major shortcomings. \textit{First}, autoregressive sampling is very inefficient because the model has to generate the final predictions step-by-step from intermediate parts. \textit{Second}, training autoregressive models requires pre-defined generation orders, while the absolute generation order of molecules is ambiguous.

To avoid these issues, a non-autoregressive model called NERF \cite{NERF} has been proposed. It predicts the redistribution of electrons, which is an important observation because chemical reactions can be reflected by bond reconnection, since the elements of the involved atoms remain unchanged. Bond reconnection inherently involves electron redistribution. For example, forming a single bond between a carbon atom and a nitrogen atom involves sharing a pair of electrons between them. Therefore, predicting electron redistribution $\Delta E$ is sufficient for inferring possible complete products. With a novel design for the electron redistribution decoder, NERF has achieved state-of-the-art top-1 accuracy and much faster sampling speed with parallel inference.

However, an important issue that has been neglected is that $\Delta \hat{E}$ must simultaneously follow the \textit{electron counting rule} and the \textit{symmetry rule} to better approximate $\Delta E$. The chemical reaction that adheres to both of these rules is illustrated in Figure~\ref{fig:electron-transfer}. We can clearly see that $\Delta E$ is row-wise conservative (conforming to the electron counting rule) and symmetric. Currently, the NERF decoder neglects to symmetrize $\Delta \hat{E}$, only ensuring that it follows the electron counting rule through combinations of \emph{BondBreaking} and \emph{BondFormation} self-attention matrices. This leads to inexact modeling of $\Delta E$ and impairs model performance. The predicted product adjacency matrix $\hat{E}^{P}$ is forced to be symmetric, but this operation implicitly corrupts the row-wise conservatism of $\Delta \hat{E}$.

In this work, we propose a novel framework called \textit{ReactionSink} to refine the decoder. The framework applies Sinkhorn's algorithm to iteratively update the \emph{BondBreaking} and \emph{BondFormation} self-attention matrices. With a sufficient number of iterations, each self-attention matrix converges to a doubly stochastic matrix. To generalize this solution to multi-head attention mechanisms, we impose augmented constraints on the Sinkhorn normalization operation.
We theoretically show that the NERF decoder leads to an implicit contradiction in preserving both physical constraints. In contrast, doubly stochastic self-attention mappings lead to valid $\Delta \hat{E}$, which preserve both rules simultaneously and are not corrupted by the symmetrization on $\hat{E}^{P}$. We further discuss the implicit connections between the reaction prediction problem and the discrete optimal transport problem to validate the rationality of imposing the doubly stochasticity constraint on the decoder as an information prior.
Comprehensive empirical results on the open benchmark dataset USPTO-MIT demonstrate that our approach consistently outperforms baseline non-autoregressive reaction prediction models. Our contributions can be summarized as follows:
\begin{itemize} 
    \item We introduce two chemical rules that the predicted electron redistribution $\Delta \hat{E}$ should follow and theoretically show that the current non-autoregressive decoder inevitably violates either of these chemical rules; 
    
    \item We propose a novel decoder design that conforms to both chemical rules through leveraging combinations of doubly stochastic self-attention mappings and we generalize our approach to multi-head attention mechanism; 
    
    \item We provide some new interpretations on electron redistribution modeling, which facilitates further research on potential connections between the reaction prediction problem and the optimal transport problem. 
\end{itemize} 

\section{Related Work}

\paragraph{Reaction Prediction.} In the early years, template-based methods were widely adopted for reaction prediction. Experts deduced possible products with the help of concluded reaction templates. Although these methods are reliable, they suffer from poor generalization. Quantum computation methods \cite{Quantum-chemistry,quantun-chem-curent-state,discover-sample} have also been adopted for reaction predictions, but they suffer from low computation speeds, despite being template-free. With the rise of deep learning, various deep learning architectures have been utilized in reaction modeling. Models that learn neural matching between reactions and templates have been proposed \cite{NeuralSym,SeglerW16}. WLDN \cite{WLDN} proposes a template-free model that splits reaction prediction into a reaction center identification stage and candidate ranking stage. MEGAN and GTPN \cite{MEGAN,GTPN} model reactions as a sequence of graph edits, while Electro \cite{elector} regards reactions as a sequence of electron transfers. A major class of models are translation-based reaction prediction models \cite{moleculartransformer,Tetko,Graph2smiles,ReactionT5,chemformer,motif-reaction}, which mainly transfer techniques from language modeling to reaction modeling. NERF \cite{NERF} refines the decoder to model electron redistribution and achieves non-autoregressive modeling with state-of-the-art top-1 accuracy and parallel inference.

\paragraph{Sinkhorn and Attention.} The Sinkhorn algorithm \cite{sinkhorn_1966,Sinkhorn-distance,computational-OT} is a well-studied tool for approximating solutions to the optimal transport problem. For instance, a new set pooling algorithm is presented in \cite{trainable-OT}, which embeds the optimal transport plan between input sets and reference sets, and then uses the Sinkhorn algorithm as a fast solver. Sinkhorn normalization is presented in \cite{manifold-bi-stochastic,spectral-diffusion}, and its convergence properties to a doubly stochastic matrix have been theoretically validated. SparseMAP \cite{sparseMAP} uses doubly stochastic attention matrices in LSTM-based encoder-decoder networks. Sinkhorn updates are also applied for differentiable ranking over internal representations in \cite{ranking-sinkhorn}. Gumbel-Sinkhorn \cite{Gumbel-Sinkhorn} applies the Sinkhorn algorithm to learn stochastic maximization over permutations. The Sinkhorn algorithm is used to iteratively generate doubly stochastic matrices for approximating smooth filters in \cite{symmetry-filter}. It has also been applied to refine self-attention mechanisms in \cite{transformer}. The Sparse Sinkhorn Transformer \cite{sparse-sinkhorn-attention} learns sparse self-attention by introducing a sorting network that generates a doubly stochastic matrix to permute input sequence elements. Sinkformer \cite{sinkformer} applies the Sinkhorn algorithm to derive doubly stochastic transformers.

\section{Proposed Methods}

\begin{figure*}[!htp]
\centering
\includegraphics[width=1.01\textwidth]{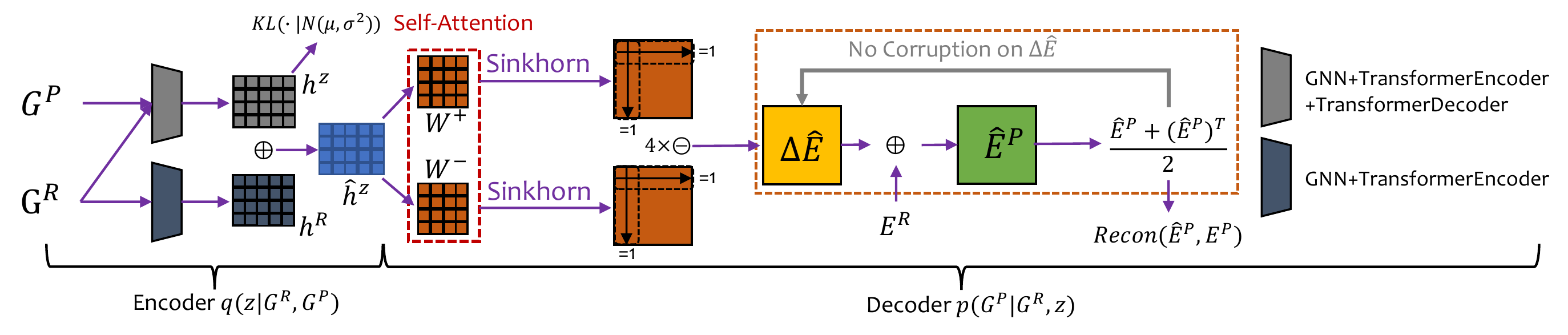}
\caption{This figure illustrates the training procedure of \textit{ReactionSink}. It adopts the conditional variational autoencoder framework with KL-divergence loss and reconstruction loss $||E^{P}-\hat{E}^{P}||_{2}^{2}$. Red rectangles highlight \emph{BondFormation} and \emph{BondBreaking} self-attention mappings. Orange rectangle highlights the process we aim to analyze. $\ominus$ denotes the element-wise subtraction, $\oplus$ denotes the element-wise plus. With Sinkhorn algorithm applied on \emph{BondFormation} and \emph{BondBreaking}, $\Delta \hat{E}$ becomes doubly conservative and would not be corrupted by symmetrization on $\hat{E}^{P}$. “4$\times$” denotes the multi-head attention with 4 heads.}
\label{fig:ReactionSink}
\end{figure*}

This section introduces the learning framework of our new non-autoregressive reaction prediction model, \textit{ReactionSink}. We first describe its problem formulation and general training process. Then we introduce its encoder architecture, which obtains the conditional reaction latent embedding. Most importantly, we present our novel decoder mechanism that leverages doubly stochastic self-attention matrices to obtain $\Delta E$. In the theoretical analysis section, we mathematically compare our decoder with the NERF \cite{NERF} decoder to demonstrate the advantages of our approach. Finally, we briefly discuss the new insights brought by \textit{ReactionSink}.

\paragraph{Problem Formulation.} A chemical reaction is represented as a pair $(G^{R}, G^{P})$, where $G^{R}=(V, E^{R}, f^{R})$ denotes the set of reactants and $G^{P}=(V, E^{P}, f^{P})$ denotes the set of products, $V$ is the set of $|V|$ atoms, $E^{R}$ and $E^{P}$ denote the adjacency matrix of $G^{R}$ and $G^{P}$ respectively, $f^{R}$ and $f^{P}$ denote the atom feature matrix of $G^{R}$ and $G^{P}$ respectively. Following the atom-mapping principle, reactants and products share the same set of atoms $V$. Unlike most of the papers representing $E$ as a 3D tensor, where $E\in \mathbb{R}^{|V|\times|V|}$ is a 2D matrix and $E_{ij}$ denotes the number of shared electrons between atom $i$ and atom $j$ (i.e. single bond is 1, double bond is 2). For special aromatic bond, we denote it as $E_{ij}=1.5$. The goal of reaction prediction is to predict the set of products $G^{P}$ given a set of reactants $G^{R}$. 

Following NERF \cite{NERF}, we also adopt the conditional variational autoencoder \cite{CVAE} (CVAE) architecture to approximate $p(G^{P}|G^{R})$ by introducing a latent variable $\mathbf{z}$. Instead of directly maximizing the log-likelihood $\log p(G^{P}|G^{R})$, CVAE is to maximize its evidence-lower bound (ELBO):
\begin{equation}
\begin{split}
\label{eq:NERF}
   ELBO= &\mathbb{E}_{q(\mathbf{z}|G^{P},G^{R})}[\log p(G^{P}|G^{R},\mathbf{z})] -\\ 
    &\hspace{-0.2in}\funcname{KL}(q(\mathbf{z}|G^{P},G^{R})||p(\mathbf{z}|G^{R}))\leq \log p(G^{P}|G^{R}) ,
\end{split}
\end{equation}
where $q(\mathbf{z}|G^{P},G^{R})$ is the reaction encoder with reaction $(G^{R},G^{P})$ as input and low-dimensional embedding $\mathbf{h}^{\mathbf{z}}$ as output, $p(G^{P}|G^{R},\mathbf{z})$ is product decoder with reactants $G^{R}$ and latent embedding $\mathbf{h}^{\mathbf{z}}$ as input, $p(\mathbf{z}|G^{R})$ denotes the prior distribution of latent variable $\mathbf{z}$. $\funcname{KL}$ term is minimizing the gap between $q(\mathbf{z}|G^{P},G^{R})$ and $p(\mathbf{z}|G^{R})$. In real implementations, the backbone network architectures of the encoder $q(\mathbf{z}|G^{P},G^{R})$ are graph neural networks (GNNs) and transformers. GNNs capture local interactions within each molecule and transformers capture global interactions across different molecules. With this architecture, $G^{R}$ and $G^{P}$ will be projected to reactants embedding $\mathbf{h}^{R}$ and products embedding $\mathbf{h}^{P}$ respectively. The cross attention layer is used for mapping $\mathbf{h}^{R}$ to latent $\mathbf{h}^{z}$ with $\mathbf{h}^{P}$ as teacher forcing during training. Finally, the latent conditional reaction embedding is derived such that $\hat{\mathbf{h}}^{\mathbf{z}}=\mathbf{h}^{R}+\mathbf{h}^{z}$. Note that similar reaction embedding networks are widely adopted in many related works as long as transformer-based reaction encoder architectures are used. Please refer the detailed architectures in Appendix 1 and the general architectures of \textit{ReactionSink} is illustrated in Figure~\ref{fig:ReactionSink}. 

\subsection{\textbf{Decoder with Physical Constraints}}

The core mechanism of the decoder $p(G^{P}|G^{R},\mathbf{z})$ is first decoding electron redistribution matrix $\Delta \hat{E}$ from conditional latent $\hat{\mathbf{h}}^{\mathbf{z}}$, and then add it to $E^{R}$ to predict $\hat{E}^{P}$, i.e., $\hat{E}^{P} = E^{R} + \Delta \hat{E}$. 
Before presenting \textit{ReactionSink} decoder in details, we first introduce two important physical rules, the \textit{electron-counting} rule and the \textit{symmetry} rule, which should be followed by the predicted $\Delta \hat{E}$.
However, these rules are not explicitly stated in the previous literature.
We will mathematically characterize these two rules as matrix properties.

\paragraph{Electron Counting Rule.} This is an important chemical rule of thumb: the octet rule. It states that main-group elements (carbon, nitrogen, oxygen, halogens) tend to bond in a way that each atom has eight electrons in its valence shell. Following this rule, electron redistribution of each atom is conservative, allowing it to maintain eight electrons in its valence shell. For example, when a carbon atom forms a new bond with one atom, it tends to break an existing bond with another atom in order to maintain eight electrons in its valence shell. This empirical rule leads to an important mathematical property of $\Delta E$, which is \textbf{row-wise conservative} (i.e., $\sum_{j}\Delta E_{ij} = 0$ for each atom $i$). Note that this rule is more informative than the charge conservation rule, which can be mathematically expressed as $\sum_{i,j}\Delta E_{ij} = 0$.

\paragraph{Symmetry Rule.} The ground-truth $\Delta E$ should clearly be symmetric. For example, if the edge number between atom $i$ and atom $j$ increases by 1, then the edge number between atom $j$ and atom $i$ should also increase by 1. To better approximate $\Delta E$, $\Delta \hat{E}$ should also be symmetric.

We have shown that $\Delta \hat{E}$ should preserve two mathematical properties which are $\sum_{j}\Delta \hat{E}_{ij} = 0$ for any $j$ and $\Delta \hat{E}_{ij} = \Delta \hat{E}_{ji}$ for any $i,j$. This can easily come up with a new conclusion that $\Delta \hat{E}$ should also be \textbf{column-wise conservative} such that $\sum_{i}\Delta \hat{E}_{ij} = 0$ for any $j$. Therefore, $\Delta \hat{E}$ should be both row-wise conservative and column-wise conservative in the strict sense, which is \textbf{doubly conservative}. Inspired by NERF decoder that leverages the combinations of row-wise stochastic self-attention mappings to achieve row-conservative $\Delta \hat{E}$, we achieve doubly conservative $\Delta \hat{E}$ by using combinations of two \textbf{doubly stochastic self-attention mappings}. Specifically, we decode two self-attention mappings \emph{BondFormation} $W^{+}$ and \emph{BondBreaking} $W^{-}$ from the conditional latent $\hat{\mathbf{h}}^{\mathbf{z}}$. \emph{BondFormation} $W^{+}$ is to predict bond addition and \emph{BondBreaking} $W^{-}$ is to predict bond removal. The detailed self-attention decoding mechanism is shown in Appendix 2. Then the combinations of $W^{+}$ and $W^{-}$ should satisfy the following equations: 
\begin{equation}
\label{eq:doublly-conservative}
\begin{split}
    (\sum_{j}W^{+}_{ij} - \sum_{j}W^{-}_{ij}) = \sum_{j}\Delta \hat{E}_{ij} = 0,\ \text{for any i,}\\
    (\sum_{i}W^{+}_{ij} - \sum_{i}W^{-}_{ij}) = \sum_{i}\Delta \hat{E}_{ij} = 0,\ \text{for any j.}
\end{split}
\end{equation}
To satisfy the above equations, we force $W^{+}$ and $W^{-}$ to be doubly stochastic such that:
\begin{equation}
\label{eq:doubly-stochastic}
\begin{split}
    \sum_{j}W^{+}_{ij} = 1, \sum_{j}W^{-}_{ij} = 1\ \text{for any i,}\\
    \sum_{i}W^{+}_{ij} = 1, \sum_{i}W^{-}_{ij} = 1\ \text{for any j.}
\end{split}
\end{equation}
Note that forcing $W^{+}$ and $W^{-}$ to be doubly stochastic is probably not the only solution to Eq~\ref{eq:doublly-conservative}. But stochasticity (sum to 1) is the most straightforward solution and more compatible with normalization operations. Therefore, if $W^{+}$ and $W^{-}$ can be doubly stochastic matrices as stated in Eq~\ref{eq:doubly-stochastic}, their combinations can lead to a doubly conservative $\Delta \hat{E}$. Then we can apply symmetrization operation on $\Delta \hat{E}$ such that $\frac{\Delta \hat{E}+(\Delta \hat{E})^{T}}{2}$, which is averaging predictions of $\Delta \hat{E}_{ij}$ and $\Delta \hat{E}_{ji}$ for each $i,j$. Interestingly, symmetrization in this case would not corrupt the conservative of $\Delta \hat{E}$. We will theoretically validate these claims in the next section. 

Unfortunately, decoding doubly stochastic self-attention mappings from conditional latent variables is non-trivial. Recall that the vanilla self-attention mapping is obtained as $\funcname{SoftMax}(\exp(QK^{T}))$, where $Q$ and $K$ are queries and keys linearly projected from latent variable $\hat{\mathbf{h}}^{\mathbf{z}}$. Note that the $\funcname{SoftMax}$ operation is basically a row-wise normalization technique that forces self-attention mappings to be row-wise stochastic, while it is not necessarily column-wise stochastic. Therefore, a specific operator is required to extend self-attention matrices, called \emph{Bondformation} $W^{+}$ and \emph{Bondbreaking} $W^{-}$, to doubly stochastic matrices. Furthermore, this operator should be differentiable so that it can be naturally merged into decoder neural networks.

\paragraph{Sinkhorn's Algorithm.} An iterative algorithm to convert positive matrix to doubly stochastic matrix is Sinkhorn's algorithm. The main idea of Sinkhorn's algorithm is alternatively applying row-wise normalization and column-wise normalization on the given positive matrix. Given a positive square matrix $X\in \mathbb{R}^{N\times N}$, we define the Sinkhorn operator as follows:
\begin{equation}
\label{eq:sinkhorn}
\begin{split}
    S^{0}(X) &= \exp(X),\\
    S^{l}(X) &= T_{c}(T_{r}(S^{l-1}(X))),\\
    S^{\infty}(X) &= \text{lim}_{l\rightarrow \infty}S^{l}(X),
\end{split}
\end{equation}
where $T_{r}$ denotes the row-wise normalization and $T_{c}$ denotes the column-wise normalization.
They are formally described as follows: 
\begin{equation}
\label{eq:norm}
\begin{split}
    T_{r}(X) &= X\oslash(X \mathbf{1}_{N} \mathbf{1}_{N}^{T}),\\
    T_{c}(X) &= X\oslash(\mathbf{1}_{N}\mathbf{1}_{N}^{T}X),
\end{split}
\end{equation}
where $\oslash$ denotes the element-wise division operator and $\mathbf{1}$ denotes the all one vector. We only need to apply the above normalization operations alternatively with enough number of iterations. In real implementation, to stabilize the updating process, we put the computation of Eq~\ref{eq:norm} in log domain. Details of this implementation are shown in the Appendix 6.

Sinkhorn \cite{sinkhorn_1966} proves that $S^{\infty}(X)$ belongs to the Birkhoff polytope, the set of doubly stochastic matrices, such that $S^{\infty}(X)\mathbf{1}_{N}=\mathbf{1}_{N}$ and $S^{\infty}(X)^{T}\mathbf{1}_{N}=\mathbf{1}_{N}$. Hence, applying this Sinkhorn operator on $W^{+}$ and $W^{-}$ separately will make them converge to doubly stochastic matrices. This Sinkhorn operator can be easily coupled with the self-attention mechanism as an additional layer. In addition, it can be perfectly suited for backpropagation updates of neural networks. Specifically, this Sinkhorn operator can be stacked after the vanilla self-attention, since the vanilla self-attention mapping can be regarded as $T_{r}(\exp(QK^{T}))$, which is equivalent to the initialization step plus the row-wise normalization step in Eq~\ref{eq:sinkhorn}. Therefore, in real implementations, we start with the column-wise normalization $T_{c}$ after the $\funcname{SoftMax}$ operation in the vanilla self-attention. 

Currently, each entry of $\Delta \hat{E}$ is in range $[-1,1]$ while edge number changes more than 1 in many cases. Therefore, we apply multi-head attention mechanism such that $\Delta \hat{E} = \sum_{d=1}^{4}W^{+d} - \sum_{d=1}^{4}W^{-d}$ with $d=4$ attention heads. In this way, $\Delta \hat{E}_{ij}$ is in range $[-4, 4]$, which could cover enough support of edge number changes of chemical reactions. We only need to manipulate Eq~\ref{eq:doubly-stochastic} a little bit by increasing conservation constraints to 4 and impose this augmented constraint on the sum of attention heads $\sum_{d=1}^{4}W^{+d}$ and $\sum_{d=1}^{4}W^{-d}$ respectively as follows: 
\begin{equation}
\begin{split}
    \sum_{d=1}^{4}\sum_{j}W^{+d}_{ij} = 4, \sum_{d=1}^{4}\sum_{j}W^{-d}_{ij} = 4\ \text{for any i,}\\
    \sum_{d=1}^{4}\sum_{i}W^{+d}_{ij} = 4, \sum_{d=1}^{4}\sum_{i}W^{-d}_{ij} = 4\ \text{for any j.}
\end{split}
\end{equation}
For extending this Sinkhorn normalization layer to the multi-head attention mechanism, we only need to impose the augmented constraints to both row-wise normalization and column-wise normalization such that: 
\begin{equation}
\begin{split}
    T_{r}(X) &= 4X\oslash(X \mathbf{1}_{N} \mathbf{1}_{N}^{T}),\\
    T_{c}(X) &= 4X\oslash(\mathbf{1}_{N}\mathbf{1}_{N}^{T}X).
\end{split}
\end{equation}
To conclude, with Sinkhorn's algorithm, we can efficiently convert $W^{+}$ and $W^{-}$ to doubly stochastic matrices, to finally achieve symmetric doubly conservative $\Delta \hat{E}$.

\section{Theoretical Analysis}

In this section, we first theoretically show that the NERF decoder cannot output valid $\Delta \hat{E}$ that follows the previously mentioned two physical constraints. In comparison, we mathematically validate that our \textit{ReactionSink} decoder can generate valid $\Delta \hat{E}$ and the symmetry is not contradictory to the conservative. 

In NERF decoder, it generates two vanilla self-attention mappings $W^{+}$ and $W^{-}$ based on conditional latent for \emph{BondFormation} and \emph{BondBreaking} respectively and approximates $\Delta \hat{E}$ through subtraction with multi-head attention such that $\Delta \hat{E} = \sum_{d=1}^{4}W^{+d} - \sum_{d=1}^{4}W^{-d}$. Then it applies symmetrization on the final predicted $\hat{E}^{P}=E^{R}+\Delta \hat{E}$ such that $\frac{E^{P}+(E^{P})^{T}}{2}$. Therefore, we begin to verify that $\Delta \hat{E}$ has also been implicitly symmetrized by NERF decoder, which is the following lemma:
\begin{lemma}
Symmetrization on the final predicted product adjacency matrix $\hat{E}^{P}$ is equivalent to implicit symmetrization on the predicted electron redistribution $\Delta \hat{E}$.
\end{lemma}
The proof procedures are shown in Appendix 2. 

Then we argue that NERF decoder design cannot satisfy the two proposed physical constraints simultaneously. Specifically, the symmetrization operation on $\hat{E}^{P}$ would result in that $\Delta \hat{E}$ violates the first rule. Since we have already shown that the symmetrization of $\hat{E}^{P}$ is equivalent to the implicit symmetrization of $\Delta \hat{E}$ in lemma~\ref{proof:symmetry-induction}, we only need to further prove the following lemma: 
\begin{lemma}
    Implicit symmetrization on $\Delta \hat{E}$ would result in that $\Delta \hat{E}$ violates the electron counting rule such that $\sum_{j}\Delta \hat{E}_{ij} \not= 0$ for some $i$. 
\end{lemma} 
The proof procedures are shown in Appendix 3. 

To summarize, electron redistribution predictions generated by NERF cannot satisfy both constraints simultaneously. Implicit symmetrization operation applied on $\Delta \hat{E}$ makes it violate the first rule. Otherwise it will violate the second rule if no symmetrization operation is adopted. In short, NERF decoder actually leads to a contradiction between two physical constraints. 

In comparison, with doubly stochastic self-attention mapping $W^{+}$ and $W^{-}$, \textit{ReactionSink} can conform to both physical rules concurrently. We can prove the following theorem:
\begin{theorem}
    If the self-attention mappings for bondformation $W^{+}$ and bondbreaking $W^{-}$ are doubly stochastic, then $\Delta \hat{E}$ would be guaranteed to be doubly conservative and its doubly conservation would not be corrupted by the symmetrization of $\hat{E}^{P}$. 
\end{theorem}
The proof procedures are shown in Appendix 4.

To extend the above conclusion to multi-head attention mechanism. The following corollary can be easily proved:
\begin{corollary}
    If the multi-head self-attention mechanisms for bondformation $\sum_{d=1}^{D}W^{+d}$ and bondbreaking $\sum_{d=1}^{D}W^{-d}$ are doubly conservative with $D$-sum constraints, then $\Delta \hat{E}$ would be guaranteed to be doubly conservative and its doubly conservation would not be corrupted by the symmetrization of $\hat{E}^{P}$. 
\end{corollary}
The proof procedures are shown in Appendix 5.

\subsection{Further Discussions}

Our novel framework, \textit{ReactionSink}, provides new insights into the non-autoregressive reaction prediction problem. Doubly stochastic matrices usually have special meaning in mathematics, with the most well-known being the permutation matrix. This suggests that electron redistribution can be formulated as a learned permutation of electrons. Note that permuting a matrix is a multiplication operation, which is different from the current formulation. The permutation formulation is intuitively closer to the discrete nature of electron redistribution, but the detailed formulation method requires further research.
More generally, reaction prediction can be connected with the optimal transport (OT) problem. The permutation problem is the most elementary discrete OT problem. Specifically, electron redistribution can also be regarded as allocating electrons to lower the energy of the entire reaction system, which is implicitly connected to the discrete OT problem. These new insights have not been revealed by existing work, and we present them here to inspire future research, although they are still in their nascent stages.  

\section{Experiments}

\begin{table*}[ht!]
    \centering
    \vskip 0.1in
    \resizebox{1.00\textwidth}{!}{%
    \begin{tabular}{c|c|c|c|c|c|c|c|c}
    \hline
       Category& Model & Top-1 & Top-2 & Top-3 & Top-5 & Top-10 & Parallel & End-to-end\\
       \hline
       Template-based & $\text{Symbolic}{\dagger}$ & 90.4 & 93.2 & 94.1 & 95.0 & - & $\checkmark$ & $\times$\\
       \hline
       Two-stage & $\text{WLDN}^{\dagger}$ & 79.6 & - & 87.7 & 89.2 & - & $\checkmark$ & $\times$ \\
       \hline
       \multirow{8}{5em}{Autoregressive} & $\text{GTPN}{\dagger}$ & 83.2 & - & 86.0 & 86.5 & - & $\times$ & $\checkmark$ \\
       & $\text{MT-base}^{\dagger}$ & 88.8 & 92.6 & 93.7 & 94.4 & 94.9 & $\times$ & $\checkmark$\\
       & $\text{MEGAN}^{\dagger}$ & 89.3 & 92.7 & 94.4 & 95.6 & 95.4 & $\times$ & $\checkmark$\\
       & $\text{MT}^{\dagger}$ & 90.4 & 93.7 & \textbf{94.6} & 95.3 & - & $\times$ & $\checkmark$ \\
       & $\text{Chemformer}^{\dagger}$ & 91.3 & - & - & 93.7 & 94.0 & $\times$ & $\checkmark$ \\
       & $\text{Sub-reaction}^{\dagger}$ & 91.0 & - & 94.5 & 95.7 & - & $\times$ & $\checkmark$\\
       & $\text{Graph2Smiles}^{\dagger}$ & 90.3 & - & 94.0 & 94.8 & 95.3 & $\times$ & $\checkmark$\\
       & $\text{AT}\times100^{\dagger}$ & 90.6 & \textbf{94.4} & - & \textbf{96.1} & - & $\times$ & $\checkmark$\\ 
       \hline
       \multirow{2}{5em}{Non-autoregressive} & NERF & 90.7 & 92.3 & 93.3 & 93.7 & 94.0 & $\checkmark$ & $\checkmark$ \\
       & ReactionSink & \textbf{91.3} & 93.3 & 94.0 & 94.5 & 94.9 & $\checkmark$ & $\checkmark$\\
       \hline
    \end{tabular}%
    }
    \caption{Top-K Accuracy \% on USPTO-479K with original random split. Best results are bolded. $\dagger$ indicates that the reported results are copied from the corresponding published papers. We also show the category of each model to reflect its property. “Parallel” indicates whether the model is capable of parallel inference. “End-to-end” indicates whether the model can make end-to-end inference.} 
    \label{tab:random-split}
\end{table*}

\label{experiments} 
\paragraph{Dataset.} Following previous work, we evaluate our approach on the open public benchmark dataset USPTO-MIT \cite{WLDN}, which contains 479K reactions filtered by removing duplicates and erroneous reactions from Lowe's original data \cite{Lowe}. An exisitng work discovers that about 0.3\% of reactions in USPTO-479K do not satisfy the non-autoregressive learning settings for various reasons~ \cite{NERF}. Following this work \cite{NERF}, we also filter 0.3\% reactions from dataset and deduct 0.3\% from the top-k accuracies of our model correspondingly. 

\paragraph{Experimental Setting.} We conduct experiments on three different splits of reaction prediction, which are random split, tanimoto-0.4 split and tanimoto-0.6 split. Random split is adopted by most of the previous work. Scaffold splits, tanimoto-0.4 split and tanimoto-0.6 split, are adopted by \cite{reactionattr} to test the generalization of reaction model with larger distribution shift between training and testing. Tanimoto similarity is measuring whether two reactions are structurally similar. Higher tanimoto index indicates that two reactions are more similar and otherwise two reactions are dissimilar with each other (e.g. Tanimoto-0.4 has larger distribution gap between training and testing). For original random split, the training set, validation set and testing set follows split ratio 409K:30K:40K. For scaffold split, the split ratio is 392K:30K:50K. 

\paragraph{Comparison Baselines.} We compare \textit{ReactionSink} with the following baseline models: \textbf{WLDN} \cite{WLDN} is a two-stage model, which firstly predicts reaction centers and then ranks enumerated products; \textbf{GTPN} \cite{GTPN} models reaction predictions as a series of graph transformations and use policy networks to learn the transformations; \textbf{MT-base} \cite{moleculartransformer} is a transformer-based autoregressive modeling with both input and output in SMILES sequence format; \textbf{MEGAN} \cite{MEGAN} models reaction predictions as series of graph edit operations and generates edit sequences in an autoregressive manner; \textbf{MT} \cite{moleculartransformer} is MT-base models with data augmentation techniques applied to SMILES sequence input; \textbf{Symbolic} \cite{symbolic} introduces the chemical rules to reaction modeling using symbolic inference; \textbf{Chemformer} \cite{chemformer} leverages molecular SMILES encoder pretrained on 100M molecular datasets with three self-supervised tasks; \textbf{Graph2Smiles} \cite{Graph2smiles} leverages the similar backbone network of \textbf{NERF} \cite{NERF}; \textbf{AT$\times$100} \cite{Tetko} leverages molecular transformer with $\times$100 SMILES augmentations; \textbf{Sub-reaction} \cite{motif-reaction} leverages motif tree to achieve substructure-aware reaction prediction; \textbf{NERF} \cite{NERF} uses CVAE to achieve non-autoregressive modeling; 

\paragraph{Model Configuration and Reproducibility Setting.} The major encoder and decoder architectures are following NERF \cite{NERF}. To ensure fair comparison with major baseline models Molecular Transformer and NERF, we set the number of transformer encoder layers and transformer decoder layers (cross-attention layer) to be 4, the same as previous work. And we set the dimension of latent embedding to be 256. For multi-head attention decoder, \emph{BondFormation} and \emph{BondBreaking} both have 4 attention heads. The model is optimized using Adam optimizer \cite{Adam} at learning rate $10^{-4}$ with linear warm-up and linear learning rate decay. The number of iterations $l$ of Sinkhorn normalization is also an hyperparameter to fine-tune. Larger $l$ indicates more rounds of Sinkhorn normalization. Finally, we train our \textit{ReactionSink} for 100 epochs with a batch size of 128 using 8 Nvidia V100 GPUs in this work.

\paragraph{Evaluation Metrics.} Following the convention of previous work, we adopt the top-k accuracies to evaluate the performance of all the compared algorithms. The top-k accuracy is the percentage of reactions that have the ground-truth product in the set of the top-k predicted molecules. As long as the set of predicted products contains the ground-truth main product, then the prediction would be counted as a correct one. In this work, the value of $k$ is set as 5 different values: $\lbrace 1,2,3,5,10 \rbrace$. 

\paragraph{Sampling Top-k Predictions.} Since \textit{ReactionSink} follows CVAE architecture, multi-modal outputs are generated through sampling $k$ different latent vectors $\hat{h}^{z}$ by increasing temperature values. To sample top-k predictions, we multiply a scalar temperature parameter $t$ to the variance of standard Gaussian distribution, such that $h^z$s are sampled from different $N(\textbf{0}, t\textbf{I})$. Increasing the temperature value $t$ would make the model predict different products. Specifically, we sample the first k predictions as our top-k predictions. Lower temperatures are set to output predictions with higher rank (e.g. prediction using $t=1$ is treated as the top-1 prediction). 
\begin{table}[t] 
    \centering
    \vskip 0.2in
    \begin{tabular}{cccc}
       \hline
       Model Name & Top-1 & Top-3 & Top-5 \\
       \hline
       WLDN5 & 75.9 & 86.2 & 88.8\\
       MT & 80.9 & \textbf{88.2} & \textbf{89.6}\\
       NERF & 85.0 & 86.7 & 88.8 \\
       ReactionSink & \textbf{86.0} & 87.2 & 89.3\\
       \hline
    \end{tabular}
    \caption{Top-K Accuracy \% on USPTO-479K with Tanimoto Similarity $<$ $0.6$ (left)}
    \label{tab:tanimoto-0.4}
\end{table}

\begin{table}[t]
    \centering
    \vskip 0.2in
    \begin{tabular}{cccc}
       \hline
       Model Name & Top-1 & Top-3 & Top-5 \\
       \hline
       WLDN5 & 69.3 & 80.9 & 84.1\\
       MT & 74.6 & 82.9 & 84.5\\
       NERF & 80.0 & 82.5 & 84.4\\
       ReactionSink & \textbf{82.2} & \textbf{83.2} & \textbf{84.5}\\
       \hline
    \end{tabular}
    \caption{Top-K Accuracy \% on USPTO-479K with Tanimoto Similarity $<$ $0.4$ (right)}
    \label{tab:tanimoto-0.6}
\end{table}

\begin{table}[t]
    \centering
    \vskip 0.2in
    \begin{tabular}{cccc}
       \hline
       Model Name & Wall-time & Latency & Speedup \\
       \hline
       Transformer (b=5) & 9min & 448ms & 1$\times$\\
       MEGAN (b=10) & 31.5min & 144ms & 0.29$\times$\\
       Symbolic & $>$7h & 1130ms & 0.02$\times$\\
       NERF & 20s & 17ms & 27$\times$\\
       ReactionSink & 24s & 20ms & 26$\times$\\
       \hline
    \end{tabular}
    \caption{Computation speedup (compared with Transformer)}
    \label{tab:speed-up}
\end{table}

\begin{table}[t]
    \centering
    \vskip 0.2in
    \begin{tabular}{cccc}
       \hline
       Model Name & Top-1 & Top-3 & Top-5 \\
       \hline
       ReactionSink (l=1) & 90.9 & 93.2 & 93.7\\
       ReactionSink (l=2) & 91.0 & 93.3 & 94.0\\
       ReactionSink (l=3) & \textbf{91.3} & 93.7 & 94.4\\
       ReactionSink (l=5) & \textbf{91.3} & 94.0 & 94.5\\
       ReactionSink (l=10) & 91.3 & 94.0 & 94.5\\
       \hline
    \end{tabular}
    \caption{Ablation studies on the number of iterations $l$ of Sinkhorn normalization}
    \label{tab:ablation-studies}
\end{table}

\paragraph{Main Experiments with Random Split and Tanimoto Splits.} We conduct experiments under random split and tanimoto splits, which can be seen as a cross-validation process. Table~\ref{tab:random-split} is the major benchmark adopted by previous work under random split conducted by WLDN \cite{WLDN}. From this table, we can see that our method reaches the state-of-the-art top-1 accuracy over all baseline methods and consistently outperforms the non-autoregressive model NERF with top-k accuracies. Note that these results are evaluated on models without knowing reagents, which means these models do not know which reactant is reagent during inference stage. We can see that currently non-autoregressive models are performing worse than autoregressive models when $k$ is higher. We conjecture that this may be caused by the simple sampling method for CVAE, or simple latent distribution assumption on uncertainty modeling (Gaussian distribution). 
In the future, it is worth to find better sampling method for CVAE, and explore stronger generative model than CVAE for reaction modeling. 

From Table~\ref{tab:tanimoto-0.4} and Table~\ref{tab:tanimoto-0.6}, we can observe that non-autoregressive models have stronger generalization than autoregressive models. We conjecture that electron redistribution modeling proposed by NERF \cite{NERF} is closer to the nature of this problem. Under tanimoto splits, \textit{ReactionSink} consistently improves the performance of non-autoregressive models in terms of top-k accuracies. 

\paragraph{Ablation Studies on Sinkhorn Iteration.} Since the Sinkhorn algorithm is also an iterative normalization process, we conduct ablation studies on the number of iterations $l$ to check its effect on model performance. From table~\ref{tab:ablation-studies}, we can see that the model performance is consistently improved when increasing $l$ from 1 to 5. However, when $l$ increases to 10, it does not bring further performance gains compared to $l=3$. This demonstrates that the \emph{Bondformation} and \emph{BondBreaking} matrices are converted to doubly stochastic matrices with few iterations of the Sinkhorn normalization. 

\paragraph{Computational Efficiency.} The main concern about \textit{ReactionSink} is its computational complexity, which is caused by the additional layers of Sinkhorn normalization. Fortunately, with a few iterations of Sinkhorn normalization, the self-attention matrices will converge to doubly stochastic matrices with negligible error. Therefore, the computational burden would not be significantly increased due to additional normalization.
Following NERF, we report the wall-time and latency of the inference model. Wall-time is the total time cost for inferring all testing samples, and latency is the inference time cost for a single testing sample. The detailed wall-time and latency computation standards are stated in NERF \cite{NERF}. From Table~\ref{tab:speed-up}, we can see that the proposed method only adds a few additional seconds to Wall-time and a few more milliseconds to latency. This demonstrates that the additional normalization layers do not create a trade-off between accuracy and speed.
Regarding the training process, additional Sinkhorn normalization only increases the training time by a few seconds for each epoch, which is almost negligible compared to the total training time. Overall, \textit{ReactionSink} is more suitable for high-throughput reaction predictions than auto-regressive methods.

\section{Conclusion}

In this work, we first introduce two chemical rules that govern electron distribution and characterize their corresponding mathematical properties. We then identify that the current electron redistribution decoder does not maintain these two physical constraints simultaneously, and we theoretically prove our claims. To address this issue, we propose the \textit{ReactionSink} architecture to extend the current self-attention mapping to doubly stochastic matrices. We also prove that the predicted electron distribution generated by our proposed methods better adheres to the physical constraints. Additionally, we establish the connections between electron redistribution and the learned optimal transport problem. This new perspective has the potential to lead to a new problem formulation for reaction prediction. Experimental results demonstrate that \textit{ReactionSink} consistently outperforms the state-of-the-art non-autoregressive reaction prediction model and does not require expensive computational costs.

\section{Acknowledgment}

The work described here was partially supported by grants from the National Key Research and Development Program of China (No. 2018AAA0100204) and from the Research Grants Council of the Hong Kong Special Administrative Region, China (CUHK 14222922, RGC GRF, No. 2151185)

\bibliographystyle{named}
\bibliography{ijcai23}

\newpage 

\section*{Appendix 1}

Each atom contains the following input features:

\begin{itemize}
    \item Atom types, an one-hot encoding indicating the atom types, such as: C, O, N, P, Cl and etc.; 
    
    \item Aromaticity, an one-hot encoding indicating the aromaticity of each atom;
    
    \item Electric charge, a 13-dimensional one-hot encoding indicating the electric charge of each atom (In range [-6, 6]); 
    
    \item Segment embedding, an one-hot encoding indicating which molecule that each atom belongs to;
    
    \item Reactant masking, an one-hot encoding indicating whether each molecule attend the reaction. This feature is only used in training and is not used in inference stage; (Reagents information.)
\end{itemize}

\paragraph{\textbf{Encoder Architecture}} Encoder architecture mainly adopts \textbf{GNN} and \textbf{transformer} to capture local and global interactions respectively. For GNN, we adopt a simple message passing mechanism by regarding each $E_{ij}$ value as weight. For transformer, we apply transformer encoder for molecular embedding and transformer decoder for applying cross attention between reactants and products. Specifically, encoding steps are as follows:
\begin{equation}
\label{eq:architecture}
\begin{split}
    \mathbf{h}^{R} &= E^{R}f^{R}, \mathbf{h}^{P} = E^{P}f^{P},\\
    \mathbf{h}^{R} &= \text{TransformerEncoder}(\mathbf{h}^{R}),\\
    \mathbf{h}^{P} &= \text{TransformerEncoder}(\mathbf{h}^{P}),\\
    \mathbf{h}^{\mathbf{z}} &= \text{TransformerDecoder}(\mathbf{h}^{R},\mathbf{h}^{P}),\\
    \hat{\mathbf{h}}^{\mathbf{z}} &= \mathbf{h}^{R} + \mathbf{h}^{\mathbf{z}},
\end{split}
\end{equation}
where $\mathbf{h}^{R},\mathbf{h}^{P},\mathbf{h}^{\mathbf{z}},\hat{\mathbf{h}}^{\mathbf{z}}\in \mathbb{R}^{|V|\times \text{dim}}$, $\text{dim}$ denotes the number of latent dimension. $\mathbf{h}^{\mathbf{z}}$ will be fed into $\text{KL}$ loss function to force its latent variable to follow standard Gaussian distribution. $\hat{\mathbf{h}}^{\mathbf{z}}$ is reactants conditional latent used for decoding. 

\section*{Appendix 2} 

\paragraph{\textbf{Self-Attention $W^{+}$ and $W^{+}$ decoding}} We decode $W^{+}$ and $W^{+}$ from conditional latent $\hat{\mathbf{h}}^{\mathbf{z}}$ through self-attention module as follows:
\begin{equation}
\label{eq:self-attention}
\begin{split}
    Q^{+} = W^{Q,+}\hat{\mathbf{h}}^{z}, K^{+} = W^{K,+}\hat{\mathbf{h}}^{\mathbf{z}},\\
    Q^{-} = W^{Q,-}\hat{\mathbf{h}}^{\mathbf{z}}, K^{-} = W^{K,-}\hat{\mathbf{h}}^{\mathbf{z}},\\
    W^{+} = \text{SoftMax}(\frac{Q^{+}(K^{+})^{T}}{\sqrt{d}}),\\
    W^{-} = \text{SoftMax}(\frac{Q^{-}(K^{-})^{T}}{\sqrt{d}}),\\
\end{split}
\end{equation}
where $Q^{+},Q^{-}$ are query matrices, $K^{+},K^{-}$ are key matrices, $W^{Q,+}, W^{Q,-}, W^{K,+}, W^{K,-}$ are trainable projection weight matrices, $\sqrt{d}$ is scaling factor where $d$ is equal to latent dimension. The above mechanism is actually same as the vanilla self-attention mechanism. 

\begin{lemma}
Symmetrization on the final predicted product adjacency matrix $\hat{E}^{P}$ is equivalent to implicit symmetrization on the predicted electron redistribution $\Delta \hat{E}$.
\end{lemma}

\begin{proof}
\label{proof:symmetry-induction}
    $Symm(\hat{E}^{P}) = \frac{\hat{E}^{P}+(\hat{E}^{P})^{T}}{2}$ 
    
    $= \frac{(E^{R}+\Delta \hat{E})+(E^{R}+\Delta \hat{E})^{T}}{2}$
    
    $= \frac{E^{R}+\Delta \hat{E}+(E^{R})^{T}+(\Delta \hat{E})^{T}}{2}$ 
    
    $= \frac{E^{R}+(E^{R})^{T}+\Delta \hat{E}+(\Delta \hat{E})^{T}}{2}$
    
    $= \frac{E^{R}+(E^{R})^{T}}{2} + \frac{\Delta \hat{E}+(\Delta \hat{E})^{T}}{2}$
    
    $= \frac{E^{R}+E^{R}}{2} + \frac{\Delta \hat{E}+(\Delta \hat{E})^{T}}{2}$ ($E^{R}$ is known to be symmetric) 
    
    $= E^{R}+\frac{\Delta \hat{E}+(\Delta \hat{E})^{T}}{2} = E^{R} + Symm(\Delta \hat{E})$
    
    Therefore, symmetrization on $\hat{E}^{P}$ in NERF is equivalent to implicit symmetrization on $\Delta \hat{E}$.
\end{proof}

\section*{Appendix 3}

\begin{lemma}
    Implicit symmetrization on $\Delta \hat{E}$ would result in that $\Delta \hat{E}$ violates the electron counting rule such that $\sum_{j}\Delta \hat{E}_{ij} \not= 0$ for some $i$. 
\end{lemma} 
\begin{proof}
\label{proof:NERF-violation}
    $\sum_{j}Symm(\Delta \hat{E})_{ij} = \sum_{j}(\frac{\Delta \hat{E}+(\Delta \hat{E})^{T}}{2})_{ij}$ (for a row with index $i$)
    
    $=\frac{1}{2}\sum_{j}((\Delta \hat{E})_{ij} + ((\Delta \hat{E})^{T})_{ij})$ 
    
    $=\frac{1}{2}\sum_{j}((\Delta \hat{E})_{ij} + (\Delta \hat{E})_{ji})$
    
    $=\frac{1}{2}(\sum_{j}(\Delta \hat{E})_{ij} + \sum_{j}(\Delta \hat{E})_{ji})$
    
    $=\frac{1}{2}(0 + \sum_{j}(\Delta \hat{E})_{ji})$ (row-wise conservative)
    
    $\not= 0$ (not column-wise conservative)
    
    Therefore, implicit symmetrization on $\Delta \hat{E}$ would corrupt the row-wise conservation of $\Delta \hat{E}$.
\end{proof}

\section*{Appendix 4}

\begin{theorem}
    If the self-attention mappings for bondformation $W^{+}$ and bondbreaking $W^{-}$ are doubly stochastic, then $\Delta \hat{E}$ would be guaranteed to be doubly conservative and its doubly conservation would not be corrupted by the symmetrization on $\hat{E}^{P}$. 
\end{theorem}

\begin{proof}
    In Lemma~\ref{proof:symmetry-induction}, we have already shown that symmetrization on $\hat{E}^{P}$ is equivalent to symmetrization on $\Delta \hat{E}$. 
    $\sum_{j}Symm(\Delta \hat{E})_{ij} = \sum_{j}(\frac{\Delta \hat{E}+(\Delta \hat{E})^{T}}{2})_{ij}$ (for a row with index $i$)
    
    $=\frac{1}{2}\sum_{j}((\Delta \hat{E})_{ij} + ((\Delta \hat{E})^{T})_{ij})$ 
    
    $=\frac{1}{2}\sum_{j}((\Delta \hat{E})_{ij} + (\Delta \hat{E})_{ji})$ 
    
    $=\frac{1}{2}(\sum_{j}(W^{+}-W^{-})_{ij} + \sum_{j}(W^{+}-W^{-})_{ji})$
    
    $=\frac{1}{2}(\sum_{j}W^{+}_{ij}-\sum_{j}W^{-}_{ij} + \sum_{j}W^{+}_{ji}-\sum_{j}W^{-}_{ji})$
    
    $=\frac{1}{2}(1-1+1-1)$ ($W^{+}$ and $W^{-}$ are doubly stochastic matrices)
    
    $= 0$ (row-wise conservative)
    
    Similarly, we can show that $\sum_{i}Symm(\Delta \hat{E})_{ij} = 0$ for any column $j$, which is column-wise conservative.
    
    Therefore, in this way, $\Delta \hat{E}$ is doubly conservative and its doubly conservative would not be corrupted by its implicit symmetrization, which is equivalent to symmetrization on $\hat{E}^{P}$. 
\end{proof}

\section*{Appendix 5}

\begin{corollary}
    If the multi-head self-attention mechanisms for bondformation $\sum_{d=1}^{D}W^{+d}$ and bondbreaking $\sum_{d=1}^{D}W^{-d}$ are doubly conservative with $D$-sum constraints, then $\Delta \hat{E}$ would be guaranteed to be doubly conservative and its doubly conservation would not be corrupted by the symmetrization on $\hat{E}^{P}$. 
\end{corollary}

\begin{proof}
    Since $\sum_{d=1}^{D}W^{+d}$ and $\sum_{d=1}^{D}W^{-d}$ are doubly conservative with D-sum constraints, such that $\sum_{j}\sum_{d=1}^{D}W^{+d}_{ij}=D$, $\sum_{i}\sum_{d=1}^{D}W^{+d}_{ij}=D$, $\sum_{j}\sum_{d=1}^{D}W^{-d}_{ij}=D$ and $\sum_{i}\sum_{d=1}^{D}W^{-d}_{ij}=D$.
    
    $\sum_{j}Symm(\Delta \hat{E})_{ij} = \sum_{j}(\frac{\Delta \hat{E}+(\Delta \hat{E})^{T}}{2})_{ij}$ (for a row with index $i$)
    
    $=\frac{1}{2}\sum_{j}((\Delta \hat{E})_{ij} + ((\Delta \hat{E})^{T})_{ij})$ 
    
    $=\frac{1}{2}\sum_{j}((\Delta \hat{E})_{ij} + (\Delta \hat{E})_{ji})$ 
    
    $=\frac{1}{2}(\sum_{j}(\Delta \hat{E})_{ij} + \sum_{j}(\Delta \hat{E})_{ji})$
    
    $=\frac{1}{2}(\sum_{j}(\sum_{d=1}^{D}W^{+d}-\sum_{d=1}^{D}W^{-d})_{ij} + \sum_{j}(\sum_{d=1}^{D}W^{+d}-\sum_{d=1}^{D}W^{-d})_{ji})$
    
    $=\frac{1}{2}(\sum_{j}\sum_{d=1}^{D}W^{+d}_{ij}-\sum_{j}\sum_{d=1}^{D}W^{-d}_{ij} + \sum_{j}\sum_{d=1}^{D}W^{+d}_{ji}-\sum_{j}\sum_{d=1}^{D}W^{-d}_{ji}$
    
    $=\frac{1}{2}(D-D+D-D)$ ($\sum_{d=1}^{D}W^{+d}$ and $\sum_{d=1}^{D}W^{-d}$ are doubly conservative with D-sum constraints.)
    
    $= 0$ (row-wise conservative)
    
    Similarly, we can show that $\sum_{i}Symm(\Delta \hat{E})_{ij} = 0$ for any column $j$, which is column-wise conservative.
    
    Therefore, in this way, $\Delta \hat{E}$ is doubly conservative and its doubly conservative would not be corrupted by its implicit symmetrization, which is equivalent to symmetrization on $\hat{E}^{P}$. Further, with this multi-head attention mechanism, $\Delta \hat{E}_{ij}$ is in range $[-D,D]$, in real implementations, $D=4$, which could cover enough support of electron changes in chemical reactions. 
\end{proof}

\section*{Appendix 6}

Given a positive matrix $S^{0}\in \mathbb{R}^{n\times n}$ such that $S^{0}=e^{X}$ for some $X\in \mathbb{R}^{n\times n}$, Sinkhorn's algorithm $(f,g)\in \mathbb{R}^{n}\times \mathbb{R}^{n}$ such that $S^{\infty}=\text{diag}(e^{f^{\infty}})S^{0}\text{diag}(e^{g^{\infty}})$ by alternatively applying row-wise and column-wise normalization in log domain, starting from $g^{0}=\mathbf{0}_{n}$,
\begin{equation}
\begin{split}
    f^{l+1}&= \log(\mathbf{1}_{n}/n)-\log(Se^{g^{l}}), \text{if}\ l\ \text{is even,}\\
    g^{l+1}&= \log(\mathbf{1}_{n}/n)-\log(S^{T}e^{f^{l}}), \text{if}\ l\ \text{is odd,}
\end{split}
\end{equation}
where $\log(Se^{g^{l}})$ and $\log(S^{T}e^{f^{l}})$ can be quickly computed by log-sum-exp. 

\end{document}